\documentclass[10pt, conference, compsocconf]{IEEEtran}
\ifCLASSINFOpdf
\else
\fi
\usepackage{bm}
\usepackage{amsmath}
\usepackage{amsfonts}
\usepackage{mathrsfs}
\usepackage{amsmath}
\DeclareMathOperator*{\argmin}{argmin}
\DeclareMathOperator*{\maximum}{\emph{max}}
\DeclareMathOperator*{\minimum}{\emph{min}}
\DeclareMathOperator*{\lexmax}{\emph{lexmax}}
\usepackage{algorithm}
\usepackage{algorithmic}
\usepackage{latexsym}
\usepackage{graphicx}
\newtheorem{lemma}{Lemma}
\newtheorem{theorem}{Theorem}
\newtheorem{proof}{Proof}
\newtheorem{definition}{Definition}

\begin{document}
%
\title{FASS: A Fairness-Aware Approach for Concurrent Service Selection with Constraints}

\author{\IEEEauthorblockN{Songyuan Li\IEEEauthorrefmark{2}, Jiwei Huang\IEEEauthorrefmark{3}\IEEEauthorrefmark{1}, Bo Cheng\IEEEauthorrefmark{2}, Lizhen Cui\IEEEauthorrefmark{4}, Yuliang Shi\IEEEauthorrefmark{4}\IEEEauthorrefmark{6}}
\IEEEauthorblockA{\IEEEauthorrefmark{2}State Key Laboratory of Networking and Switching Technology\\
Beijing University of Posts and Telecommunications,
Beijing 100876, China}
\IEEEauthorblockA{\IEEEauthorrefmark{3}Department of Computer Science and Technology,
China University of Petroleum - Beijing,
Beijing 102249, China}
\IEEEauthorblockA{\IEEEauthorrefmark{4}School of Software, Shandong University, Jinan 250101, China}
\IEEEauthorblockA{\IEEEauthorrefmark{6}Dareway Software Co., Ltd., Jinan 250101, China}
Email: lisy@bupt.edu.cn, huangjw@cup.edu.cn, chengbo@bupt.edu.cn, clz@sdu.edu.cn, shiyuliang@sdu.edu.cn
}
%
\maketitle

\begin{abstract}
The increasing momentum of service-oriented architecture has led to the emergence of divergent delivered services, where service selection is meritedly required to obtain the target service fulfilling the requirements from both users and service providers. Despite many existing works have extensively handled the issue of service selection, it remains an open question in the case where requests from multiple users are performed simultaneously by a certain set of shared candidate services. Meanwhile, there exist some constraints enforced on the context of service selection, e.g. service placement location and contracts between users and service providers. In this paper, we focus on the QoS-aware service selection with constraints from a fairness aspect, with the objective of achieving max-min fairness across multiple service requests sharing candidate service sets. To be more specific, we study the problem of fairly selecting services from shared candidate sets while service providers are self-motivated to offer better services with higher QoS values. We formulate this problem as a lexicographical maximization problem, which is far from trivial to deal with practically due to its inherently multi-objective and discrete nature. A fairness-aware algorithm for concurrent service selection (FASS) is proposed, whose basic idea is to iteratively solve the single-objective subproblems by transforming them into linear programming problems. Experimental results based on real-world datasets also validate the effectiveness and practicality of our proposed approach.

\end{abstract}

\begin{IEEEkeywords}
service selection; Quality of Service (QoS); max-min fairness; selection constraints; concurrent service execution.
\end{IEEEkeywords}

%
\IEEEpeerreviewmaketitle
\vspace{-2ex}
\section{Introduction}
Nowadays, service selection has become a key building block of Service-Oriented Architecture (SOA) along with the prevalence of services computing technology. It implies the process of gaining target service from various candidate services, whose objective is to match both functional and non-functional requirements. With the increasing scale of web services, candidate services with equivalent functionality are simultaneously provided for selection, but vary in non-functional properties \cite{Qi2010Combining}. Non-functional properties evaluates how well a service will serve for the user, which is usually represented by Quality of Service (QoS).

The common goal of QoS-aware service selection is to elect the target service with the optimal end-to-end QoS, which is inherently an optimization problem \cite{Zeng2004QoS}. There have been a great number of existing works proposing efficient service selection schemes, especially for web and cloud systems. While most existing work in the literatures primarily deals with finding the single target service from candidate services for one user, however, little focus has been on the service selection scenario with multiple service requests addressed by users simultaneously. Multiple service requests submitted by users are required for concurrent service running at the service platform. In other words, the procedure of service selection for each service request should be synchronously conducted, and all the requests from all users should be considered at the same time in service selection problem.

For this case, service requests proposed by divergent users may have various constraints. For example, when mobile communication users request for establishing links with the base station (BS), there have been the selection rule (e.g. location-aware \cite{Chen2015Dynamic}) restricting the range of deliverable BS. In the fields of content distribution, content users attributable to multiple Internet Service Providers (ISPs) have hard constraints about the Content Distribution Netoworks (CDNs) that they can access to \cite{Adhikari2012Unreeling}. Besides, users and service providers reach an agreement in contract, specifying that users can merely use the paid services. Therefore, the constraints should be fully considered especially for the cases in concurrent service selection.

Furthermore, multiple service requests may share the limited amount of candidate services with the identical functional capacities but different QoS levels. In this situation, multiple service requests are inherently competing for the candidate services with each other for the purpose of obtaining a higher QoS. Therefore, it necessitates a fairness-aware selection mechanism when pooling a shared set of candidate services.

In this paper, we put forward a fairness-aware service selection scheme, addressing the problem of multiple QoS-aware service selection with constraints. Our service selection approach is carefully designed from the perspective of service ecosystem \cite{Castelli2015Engineering} with a top-down viewpoint. On the one hand, users are usually willing to gain a better service with a higher QoS at a reasonable price. Our proposed approach fully takes care of the QoS with respect to each individual service request, and encourages each of them to acquire the target service with high and acceptably fair QoS. On the other hand, each service request gains a better service with a higher QoS without degrading the QoS of other service requests, which ensures the fairness of concurrent service selection. It is helpful for holding all the existing users in the ecosystem and attracting more users from the outside with the fair policy. With the growing scale of users, service providers will gain more revenue motivating them to develop services with higher QoS. In this way, the loop of sustainable SOA development is built up.

Highlights of our original contributions are as follows. We firstly outline our basic model of multiple service selection with constraints and formulate the max-min fairness (MMF) optimization objective as a lexicographical maximization problem. In virtue of the multi-objective and discrete characteristics of the lexicographical maximization problem, it is often a multifaceted and untractable puzzle to work out the explicit exact solution. Through extensively investigating the structure of lexicographical problem, we find out two kind properties which are separable convex objective and totally unimodular linear constraints respectively. Thanks to these two properties, we transform the lexicographical maximization problem into a range of  equivalent linear programming (LP) subproblems. The target services for multiple service requests through finite iterations of LP, where max-min fairness is achieved. Finally, the efficiency and practically of our proposed approach is validated through experiments based on the real-world dataset.
\vspace{-1.25ex}
\section{Related Work}
\begin{table}
\centering
\caption{Summary of Notations}\label{tab:notations}
\begin{tabular}{|c||p{6.6cm}|}
	\hline
	Notation & \multicolumn{1}{c|}{Definition} \\\hline
	\hline
	$n, N, \mathcal{N}$ & index, number, set of multiple service requests (users) \\\hline
	$i, M, \mathcal{M}$ & index, number, set of third-party service providers \\\hline
	$C_i$ & set of candidate services offered by provider $i$ \\\hline
	$j$ & index of service from the candidate set $C_i$ \\\hline
	$S_n$ & set of service providers for service request $n$ authorized to select service from \\\hline
	$\mathcal{N}_i$ &  set of service requests which are authorized service access by the service provider $i$\\\hline
	$x^n_{i,j}$ & whether the service request $n$ elect the $j^{th}$ candidate service in the candidate set $C_i$ (=1) or not (=0) \\\hline
	\bm{$\Theta$} & solution space formed by all of variables $x_{i,j}^n$, $\forall n\in \mathcal{N},\, i\in S_n,\, j \in C_i$  with all possible values of \{0, 1\} \\\hline
	$Q_{i,j}$ & QoS value of the $j^{th}$ service in the candidate set $C_i$ \\\hline
	$Q^{(ref)}_n$ & reference QoS value of service request $n$ \\\hline
	$\tau_n$ & execution time of service request $n$ \\\hline
    $\pi^n_{i,j}$ & user $n$'s payment when selecting the $j^{th}$ candidate service from the candidate set $C_i$  \\\hline
    $\pi_n$ & user $n$'s overall payment \\\hline
	$a_n$ & minimum payment for launching the service request $n$ \\\hline
	$b_n$ & maximum extra bonus for obtaining a better service outperforming the QoS baseline $Q^{(ref)}_n$ \\\hline
	$\bm{\varpi}$ & payment vector formed by the terms $\pi^n_{i,j}$, $\forall n\in \mathcal{N}, i \in S_n, j \in C_i$ \\\hline
	$k, K, \bm{\varpi}_k$ & index, length, the $k^{th}$ term of vector $\bm{\varpi}$ \\\hline
	$\lambda_{i,j}^{n,0}, \lambda_{i,j}^{n,1}$ & real variables generated by $\lambda$-technique \cite{Meyer1976A} in LP transformation \\\hline
	\hline
\end{tabular}
\vspace{-0.5cm}
\end{table}
\vspace{-0.5ex}
\subsection{Service Selection based on Single Service Request}
\vspace{-0.5ex}
Given the increasing number of service providers and diversified types of services they delivered on varying quality offering and pricing strategies, service selection is a highly indispensable research issue which has been intensively explored. Since service selection is conducted to pick out and allocate the optimal service to the user, it is intuitive that service selection is modeled as a mixed integer linear programming problem (MILP) which is NP-hard.

The heuristic-driven approach is usually applied to reduce the computational cost of finding the optimal service, whose weakness resides in simply providing a sub-optimal selection closer to the optimum one \cite{Casalicchio2009Optimal}. Addressing the issue of service selection in mobile edge computing environment, Wu et al. \cite{Wu2018Service} put forward a selection scheme minimizing the response time, whose arithmetic design integrated genetic and simulated annealing algorithm. Powell et al. \cite{Powell2018Optimal} conducted research on service selection in the case study of cloud resource configurations, structuring the set of feasible configuration from which diverse Pareto-optimal configurations were selected.

Furthermore, the machine learning (ML) methodology is also taken into use, intelligentizing the procedure of service selection \cite{Kirchner2015Classification}. On the basis of the reputation-driven matrix factorization approach, Xu et al. \cite{Xu2016Web} conducted accurate QoS prediction on unknown web services, supporting the context of service selection. Saleem et al. \cite{Saleem2017Personalized} designed a ML-based service selection approach based on learning-to-rank algorithm, taking advantage of historical service-selection decisions/outcomes and eventually delivering personalized service selection of each user.
\vspace{-1ex}
\subsection{Concurrent Service Selection Across Multiple Service Requests.}
\vspace{-0.75ex}
Although significant attention has been received in service selection, few selection algorithms have been proposed when multiple service requests are in the service-oriented system for concurrent execution, no wonder from the point of fairness. Bi-criteria \cite{Bessai2013Bi} analyzed concurrent sequential workflows from two aspects which are makespan and execution cost, and brought forward a fairness-aware service selection policy, preventing the running service of workflow from prolonged idleness and starvation. Specifically, the less a running service is affected by others fighting for simultaneous execution with shared candidates, the higher priority it attains to be selected to be performed in advance. In this way, all of concurrent workflows are well-done to be executed without over penalized. Services with highest execution priority make up the Pareto Set, solved by heuristic algorithms.

Regarding the fairness metrics, max-min fairness has been thoroughly studied in the issue of resource allocation across the fields of distributed systems, data center network  and queueing systems \cite{Ghodsi2011Dominant}. Ghodsi et al. \cite{Ghodsi2013Choosy} generalized the max-min fairness to Dominant Resource Fairness on account of resource allocation with placement constraints, where several properties like sharing incentive and strategy-proofness got theoretically proved. Saha et al. \cite{Saha2018Exploring} investigated the practical factors in Mesos agents which prevented itself from the desired DRF allocation and proposed specified configuration suggestions, filling the gap between target DRF resource allocation and actual resource allocation in the multi-user Apache Mesos cluster.

To overcome these limitations mentioned early, we put forward a novel approach for concurrent service selection, where the criteria of fairness is taken into account. To our best knowledge, it is the first time to introduce the concept of max-min fairness into the domain of service-oriented computing.
\vspace{-1.75ex}
\section{Problem Formulation}
\begin{figure}
\centering
\includegraphics[width=3.33in, keepaspectratio]{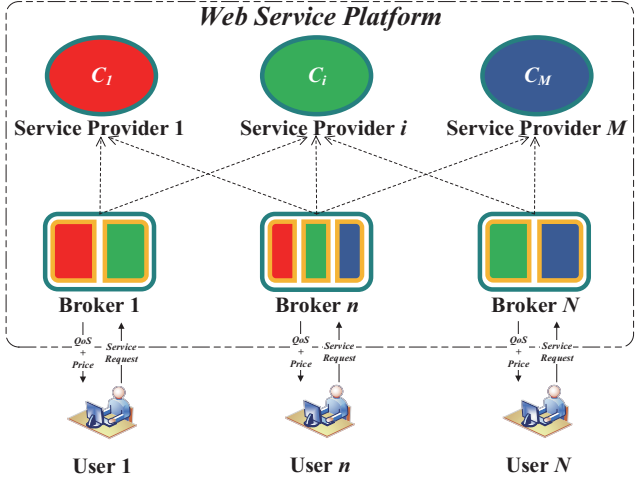}\\
\vspace{-1.75ex}
\caption{Overview of Service Selection Model.}
\label{fig:model}
\vspace{-0.6cm}
\end{figure}
\vspace{-0.75ex}
In this section, we firstly provide the model description towards multiple service selection with constraints, based on which a lexicographical problem achieving max-min fairness is delicately formulated.
\vspace{-1.75ex}
\subsection{Service Selection Model with Constraints for Multiple Service Requests}
We consider a set of multiple service requests $\mathcal{N} = \{1,2,...,N\}$ submitted by users to a web service platform for concurrent execution, as depicted in Fig. \ref{fig:model}. Substaintial candidate services are released by various third-party service providers. Given the QoS preference for each service request, the service broker is responsible for finding out the personalized target service from the numerous released services.

Without loss of generality, it is assumed that the candidate service sets from $M$ third-party service providers can be categorised into $\mathcal{M} = \{C_1,C_2,...C_M\}$. The candidate set contains a variety of services $j \in C_i$, where $1 \le i \le M$. As discussed above, the candidate set are sharable with constraints amongst multiple service request. The service selection constraint for service request $n$ is indicated by the constraint set $S_n$. The element $i \in S_n$ implies the enabled types of services which service request $n$ can elect. From the standpoint of service providers, the set of service requests authorized by provider $i$ is characterized with $\mathcal{N}_i$.

Response time, also called service execution time, is one of the most important QoS criteria in the service-oriented literature. For simplicity, response time is applied as the only QoS criteria in this paper. For each service $j$ in $C_i$, the response time is measured as the value of $Q_{i,j}$.

The selection of candidate service is formulated by a binary variable $x^n_{i,j}$, where 1 means the $j^{th}$ service in the candidate set $C_i$ is elected by the service request $n$ and 0 indicates the opposite. Decision variables of the service request $n$ is represented by $\bm{x}_n = \{x^n_{i,j} | i \in S_n, j \in C_i\}$, and all of variables $x^n_{i,j}$ forms the solution space \bm{$\Theta$}. Thus, the execution time $\tau_n$ for service request $n$ can be represented as the following equation.

\begin{equation} \label{eq:exeTm}
    \tau_n = \sum\limits_{i \in S_n}\sum\limits_{j \in C_i} x_{i,j}^n Q_{i,j}
\end{equation}

Given diverse QoS requirements from users, a tailored Service Level Agreement (SLA) is highly required for a flexible service selection scheme to propose a satisfying service assignment plan \cite{Yan2007Autonomous}. To be more specific, an SLA is defined by the QoS committed by the service provider and associated payment which the user is obliged to afford. In this work, we assume that the user pays for the service in the pattern of pay-per-use. The pricing model of pay-per-use has been widely accepted in the field of cloud service \cite{EC2}, and so is in the case for service-oriented computing \cite{Kofler2010User}. Customers wish to be served by a better service with a higher QoS even though they are reasonably asked for more money. In the pay-per-use model, the payment for service request $n$ mainly consists of two parts, one of which is the basic payment $a_n$ for launching the service which the another is the maximum extra bonus $b_n$ for delivering a better service . A baseline of QoS criteria (i.e. response time) $Q_n^{(ref)}$ is addressed here, reflecting the user $n$' basic QoS requirements. If user $n$ obtains a service outperforming
the QoS baseline $Q_n^{(ref)}$, then a basic payment $a_n$ and an extra bonus should be charged. Otherwise, the user $n$ will pay at most $a_n$ without any extra bonus. Therefore, the user $n$'s payment is calculated as (\ref{eq:pay_func}) when selecting the $j^{th}$ candidate service from the candidate set $C_i$.
\vspace{-1ex}
\begin{equation} \label{eq:pay_func}
	 \pi_{i,j}^n = a_n + b_n  \cdot (1 - \frac{Q_{i,j}}{Q_n^{(ref)}} \cdot x_{i,j}^n)
\end{equation}
where $a_n$ and $b_n$ incorporates the pricing policy after joint negotiation between user $n$ and service providers. Extra bonus is widely acceptable to be linearly increasing with the QoS improvements \cite{Irwin2004Balancing}. In general, basic payment $a_n$ is positively correlated to the severity of QoS requirements (i.e. QoS baseline $Q_n^{(ref)}$). The user $n$'s overall payment is formulated by (\ref{eq:overall_pay_n}).
\vspace{-1.5ex}
\begin{equation} \label{eq:overall_pay_n}
\pi_n =  a_n + b_n  \cdot (1 - \sum\limits_{i \in S_n}\sum\limits_{j \in C_i} \frac{Q_{i,j}}{Q_n^{(ref)}} \cdot x_{i,j}^n)
\end{equation}
\vspace{-1.5ex}
\subsection{Lexicographical Problem Formulation Achieving MMF}
\vspace{-0.5ex}
Since the candidate services released by service providers are shared by multiple service requests waiting for concurrent execution, our design purpose is to take each service request into consideration and motivate all of them to obtain the target service with high and acceptably fair QoS. To be more specific, our service selection scheme applies max-min fairness across multiple service requests. The definition of max-min fairness is given as Definition \ref{def:mmf}.
\vspace{0.5ex}
\begin{definition}[Max-Min Fairness]\label{def:mmf} A service selection scheme satisfies max-min fairmess (MMF), if it is impossible to increase the $i^{th}$ lowest payment across $N$ service requests even though removing the service requests whose payment is strictly higher than the $i^{th}$ lowest payment, note that $i \in \mathcal{N}$.
\end{definition}
\vspace{0.5ex}

By the definition of MMF, we seeks to
maximize the lowest payment amongst the multiple requests,
then to optimize the second lowest without impacting the
previous one, and so forth. Until all of the service requests
have been optimized, the procedure of service selection will
be terminated with an MMF service selection scheme obtained.

In the area of multi-criteria optimization, lexicographical techniques \cite{Zykina2004} grants the highest optimization priority to the most important objective, matching the interests of max-min fairness. As a result, our service selection scheme based on max-min fairness can be rigorously formulated as a lexicographical maximization problem, theoretically defined as the objective function ({\ref{eq:lex_obj}}) subject to the constraint equations (\ref{eq:customer_constraint}) - (\ref{eq:binary_def}). In the scenario of our work, there exist two main types of constraints which are user constraints and provider constraints. The user constraints (\ref{eq:customer_constraint}) ensure that each customer's request should elect just only one service from available candidates of her own. The provider constraints (\ref{eq:provider_constraint}) imply that different user has to select different services from service providers.
\vspace{-1.5ex}
\begin{equation} \label{eq:lex_obj}
    \lexmax\limits_{x_{i,j}^n \in \,\bm{\Theta}}\;\; \bm{\pi} = (\pi_1, \pi_2, ... ..., \pi_N)
    \vspace{-1ex}
\end{equation}
\vspace{-1.5ex}
subject to,
\begin{equation} \label{eq:customer_constraint}
    \sum\limits_{i\in S_n}\sum\limits_{j\in C_i}x_{i,j}^n = 1,\;\forall n \in \mathcal{N}
\end{equation}
\begin{equation} \label{eq:provider_constraint}
    \sum\limits_{n\in \mathcal{N}_i}x_{i,j}^n \le 1,\;\forall i \in \mathcal{M},\,\forall j \in C_i
\end{equation}
\begin{equation} \label{eq:binary_def}
    x_{i,j}^n\, \in\, \{0,1\},\;\forall n\in \mathcal{N},\, i\in S_n,\, j \in C_i
\end{equation}

The objective in the lexicographical maximization problem is a payment vector $\bm{\pi} \in \mathbb{R}^N$, each element of which represents the payment of a specified user submitting the service request $n$. Optimal $\bm{\pi}^*$ is lexicographically no smaller than any feasible $\bm{\pi}$. It signifies that the first smallest element of $\bm{\pi}^*$ (i.e., the lowest payment across multiple requests) should be the maximum amongst all feasible solutions $\bm{\pi}$. In the case of all $\bm{\pi}$ with the same lowest payment, the second lowest payment in $\bm{\pi}^*$ is applied for maximization. The rest is in a similar fashion. Through solving this lexicographical problem iteratively, an optimal service selection plan is worked out, jointly maximizing the payment of each service request and achieving the max-min fairness.
\vspace{-1.2ex}
\section{Computing Service Selection Plan}
\vspace{-0.5ex}
In this section, we explore how to compute the max-min-fairness service selection plan. Introducing the $\lambda-$technique \cite{Meyer1976A} and linear relaxation, an equivalent LP transformation is conducted, which significantly helps for the improvement of algorithm efficiency.
\vspace{-1.6ex}
\subsection{Iterative MMF Optimization Framework}
\vspace{-0.5ex}
An iterative MMF optimization framework  namely \emph{FASS} is put forward in the first step. Both payment parameters and QoS baselines are tracked for each service request. The service assignments for all $N$ service requests are iteratively accomplished one after another according to the non-decreasingly order of service payments. In the first round of iterations, the service request $n^*$ with the lowest payment is prioritized for service selection and payment optimization, treated as a subproblem implemented by a Linear Programming (LP) problem in the Section \ref{subsec:lp_trans}.

Once the candidate service optimizing the service request $n^*$'s payment is picked out, there are several settings ready for the next iteration round. In brief, we freeze the service assignment of optimized request $n^*$. First, the family of decision variables $\{x_{i,j}^n\,|\,n=n^*\}$ holds as unchanged, and lowers the dimension of the solution space \bm{$\Theta$} by one. Second, the solution space \bm{$\Theta$} should be also reduced by the decision variables relevant to the selected candidate services, formulated by $\{\,x_{i,j}^n\,|\,x_{i,j}^{n^*} = 1, i \in S_{n^*}, j \in C_i\}$. After the service request with the lowest payment having been optimized, the next round is launched aimed to optimizing the service request with the second lowest payment. Preparing for the afterwards round, we conduct the settings of solution space \bm{$\Theta$} analogous to what is done at the first round.

Such iterative process repeats until all of the service requests obtain the target service of her own. The iteration algorithm terminates, indicating the arise of concurrent service selection scheme with max-min fairness. It should be noticed that the service selection scheme is obtained through deterministic finite iterative rounds. Algorithm \ref{alg:mmf} illustrates pseudo-code for concurrent service selection achieving max-min fairness (FASS).

So far, the MMF optimization framework for concurrent service selection has been comprehensively provided. By means of equivalent LP transformation, Section \ref{subsec:lp_trans} will intensively tackle the challenges incurred by the discrete and multi-objective nature of original lexicographical problem.

\begin{algorithm}[t]
	\caption{FASS: Service Selection across Multiple Requests with Max-Min Fairness.}
	\label{alg:mmf}
	 {\bf Input:}
	Basic Payment $\mathcal{A}$, Extra Bonus $\mathcal{B}$, and QoS Baseline $\mathcal{Q}^{(ref)}$; \\
	{\bf Output:}
	Service Assignment $x_{i,j}^n, \forall n \in \mathcal{N}, i \in S_n, j \in C_i;$
	\begin{algorithmic}[1] \label{alg:mmf}
		\STATE Initialize $\widetilde{\mathcal{N}} \gets \mathcal{N};$
		\vspace*{0.065in}\WHILE {$\widetilde{\mathcal{N}} \neq \emptyset$}
		\vspace*{0.065in}\STATE $\textbf{\emph{x}} \gets LP(\mathcal{A},\mathcal{B}, \mathcal{Q}^{(ref)},\mathcal{Q},\Theta)$; \hspace{0.1in} \small{$\rhd$ Solve the LP problem (\ref{eq:obj-lp}) to obtain the scheduling plan \textbf{\emph{x}}}
		\vspace*{0.065in}\STATE \normalsize{\textbf{\emph{x$_{n^*}$}} $\gets \argmin\limits_{n \,\in \; \mathcal{N}}\,\pi_n$ \,;}\hspace{0.1in} \small{$\rhd$ Obtain workflow $n^*$'\,s optimal scheduling plan}
		\vspace*{0.065in}\STATE \normalsize{Fix the variable subset \textbf{\emph{x$_{n^*}$}}\,;}
		\vspace*{0.065in}\STATE \normalsize{Set $x_{i,j}^n \gets 0,$ in the case of arbitrary $ n \neq n^*$\,;}
		\vspace*{0.065in}\STATE \normalsize{$\Theta \gets\, \Theta \backslash \{x_{i,j}^n\,|\,n=n^*\}$} \hspace{0.035in} \small{$\rhd$ Lower the dimension of solution space $\Theta$ by one\,}
		\vspace*{0.065in}\STATE \normalsize{$\Theta \gets\, \Theta \cap\,\{\,x_{i,j}^n\,|\,x_{i,j}^{n^*} = 1, i \in S_{n^*}, j \in C_i\}\,$;} \hspace{0.035in} \small{$\rhd$ Reduce the solution space $\Theta$\,}
		\vspace*{0.065in}\STATE \normalsize{$\widetilde{\mathcal{N}} \gets \widetilde{\mathcal{N}}\backslash\{n^*\}$;} \hspace{0.1in} \small{$\rhd$ Hold all but $n^*$ and prepare for the next round}
		\vspace*{0.065in}\ENDWHILE
		\vspace*{0.065in}\RETURN $x_{i,j}^n, \forall n\in \mathcal{N},\, i\in S_n,\, j \in C_i;$
	\end{algorithmic}
\end{algorithm}
\vspace{-1ex}
\subsection{LP Transformation Towards the Lowest Payment Maximization}\label{subsec:lp_trans}
\vspace{-0.5ex}
The lexicographical optimization problem (\ref{eq:lex_obj}) is an integer optimization with multi-objectives, which is NP-hard to solve the problem directly. Given that, we resolve the problem (\ref{eq:lex_obj}) into $N$ single-objective subproblems optimizing the lowest payment. The optimization goal of single-objective subproblem is formulated as the equation (\ref{eq:single-obj}).\vspace{-1ex}
\begin{equation}\label{eq:single-obj}
\maximum\limits_{x_{i,j}^n \in \,\bm{\Theta}} \ \minimum\limits_{n \in \mathcal{N}}\,( a_n + b_n \times (1 - \frac{\tau_n}{Q_n^{(ref)}}) )
\end{equation}

Thanks to the possible value for the decision variable $x_{i,j}^n$ confined to $\{0,1\}$, the execution time $\tau_n$ for service request $n$, previously formulated by the equation (\ref{eq:exeTm}), can be also expressed as the equation (\ref{eq:exeTm_maxForm}).\vspace{-1.25ex}
\begin{equation}\label{eq:exeTm_maxForm}
\tau_n = \maximum\limits_{i \in S_n,j \in C_i} x_{i,j}^n Q_{i,j}
\end{equation}
Then, let the equation (\ref{eq:exeTm_maxForm}) substituted into the objective function (\ref{eq:single-obj}), then we have the single-objective problem represented in another non-linear form (\ref{eq:single-obj-substituted}), subject to the constraints (\ref{eq:customer_constraint})-(\ref{eq:binary_def}).
\vspace{-2ex}
\begin{equation}\label{eq:single-obj-substituted}
\maximum\limits_{x_{i,j}^n \in \,\bm{\Theta}} \ \minimum\limits_{n \in \mathcal{N},i \in S_n,j \in C_i} a_n + b_n \cdot (1 - \frac{ Q_{i,j}}{Q_n^{(ref)}} \cdot x_{i,j}^n)
\end{equation}

\textbf{Integral Optimum Guarantee.} A linear programming problem will yield an optimal solution in integers, if it has a totally unimodular (TU) coefficient matrix \cite{Korte2006Combinatorial}. In our problem domain, the coefficient matrix of constraints (\ref{eq:customer_constraint}) and (\ref{eq:provider_constraint}) is carefully investigated and verified the property of total unimodularity, in order to further determine whether the deletion of the integrality constriants (\ref{eq:binary_def}) impacts on the optimal service selection of the problem (\ref{eq:lex_obj}).

\begin{lemma}\label{lem:tu}
The matrix formed by the coefficients of constraints (\ref{eq:customer_constraint}) and (\ref{eq:provider_constraint}) is total unimodular.
\end{lemma}

\begin{proof}
Suppose that the matrix $\textbf{A}_{s\times t}$ represents the coefficient matrix form by constraints (\ref{eq:customer_constraint}) and (\ref{eq:provider_constraint}). The number of rows $s = N+\sum_{i\in\mathcal{M}} \vert C_i \vert$ indicates the total amount of both customer and provider constraints while the number of columns $t = \sum_{n \in \mathcal{N}} \sum_{i \in S_n} \vert C_i \vert$ characterizes the dimension of decision variable $\textbf{\emph{x}}$.

The matrix $\textbf{A}_{s\times t}$ is sufficiently judged as a totally unimodular matrix if it reaches two conditions below. In view of $x_{i,j}^n\, \in\, \{0,1\}$, it is easily seen that the coefficient matrix $\textbf{A}_{s\times t}$ satisfies the first condition which claims that all entries is 0 or $\pm1$. In terms of the second condition, we elect the entries of rows which belongs to the row subset $\{1,2,...,N\}$ to compose the set $R_1$, and the remaining entries of rows form up the set $R_2$, satisfying $R_1 \cap R_2 = \emptyset$. Given the constraints (\ref{eq:customer_constraint}), it is obvious to point out that, the summation of entries grouped by columns in rows $R_1$ is a $1 \times t$ vector with all the entries equal to 1. Regarding the constraints (\ref{eq:provider_constraint}), similarly, the summation of entries grouped by columns in rows $R_2$ is also a $1 \times t$ vector whose entries are 1. Therefore, it can be concluded that $\sum_{i_1 \in R_1} a_{i_1 j} \le 1$ and $\sum_{i_2 \in R_2} a_{i_2 j} \le 1$ for $\forall j \in \{1,2,...,t\}$, further satisfying the second condition  $\vert \sum_{i_1 \in R_1} a_{i_1 j} - \sum_{i_2 \in R_2} a_{i_2 j} \vert \le 1,$ for $\forall j \in \{1,2,...,t\}$.
\end{proof}
\vspace{0.3ex}
From Lemma \ref{lem:tu}, it follows that our problem has an integral optimum as long as any optimum exists, providing the legality basis of linear relaxation on the integer constraints (\ref{eq:binary_def}). The integer constraints' relaxation is formulated as follows.\vspace{-2ex}
\begin{equation}\label{eq:rational_def}
    x_{i,j}^n\, \in\, \mathbb{R}^+,\;\forall n\in \mathcal{N},\, i\in S_n,\, j \in C_i
\end{equation}
\textbf{Equivalent Convex Objective.} The optimal service selection scheme of the problem (\ref{eq:single-obj-substituted}) can be attained by solving the following lexicographical problem as (\ref{eq:single-obj-lex}). The common goal of this problem is to maximize the lowest payment across multiple service request, which is specifically the minimum element in $\bm{\varpi}$. Thus, it shows that the optimal decision variable $\bm{x}^*$ derived from the problem (\ref{eq:single-obj-lex}) is equivalent to the optimal solution of the problem (\ref{eq:single-obj-substituted}).
\vspace{-1.5ex}
\begin{equation}\label{eq:single-obj-lex}
\lexmax\limits_{x_{i,j}^n \in \,\bm{\Theta}} \quad \bm{\varpi} = (\pi_{i,j}^n\,|\,n \in \mathcal{N}, i \in S_n, j \in C_i)
\end{equation}
\vspace{-0.1ex}
In order to eventually supply a linear objective function, a tailored separable convex objective function $\xi(\bm{\varpi})$ is elaborately defined as follows, served as an intermediate transformation of objective function. The $k^{th}$ element of \,$\bm{\varpi}$ is labeled by $\varpi_k$.
\vspace{-1.5ex}
\begin{equation}\label{eq:separable-obj}
\xi(\bm{\varpi}) = \sum\limits_{k = 1}^{\vert \bm{\varpi} \vert} {\vert \bm{\varpi} \vert}^{- \varpi_k}= \sum\limits_{k = 1}^K K^{- \varpi_k}
\end{equation}
\vspace{-2.5ex}
\begin{lemma}\label{lem:separable-obj}
$\xi(\cdot)$ reverses the original partial order of lexicographically no greater than $(\succeq)$, which is mathematically represented as $\bm{\varpi}(\bm{x}^*) \succeq \bm{\varpi}(\bm{x})$ $\Leftrightarrow \xi(\bm{\varpi}(\bm{x}^*)) \le \bm{\varpi}(\bm{x})$.
\end{lemma}

\begin{proof} Due to analogous proof presented in \cite{Chen2017Scheduling}, we simplify the proof here. Suppose that $\bm{g}, \bm{\rho}$ meet the requirement of $\bm{g} \prec \bm{\rho}$. If the integer $\tilde{k}$ indicates the first non-zero element of $\langle\bm{\rho}\rangle - \langle\bm{g}\rangle$, then we have $\langle\bm{g}\rangle_k = \langle\bm{\rho}\rangle_k$ $  (\forall k \in \{1,...,\tilde{k}-1\})$ and  $\langle\bm{g}\rangle_{\tilde{k}} = \langle\bm{\rho}\rangle_{\tilde{k}}$. Let $\langle\bm{g}\rangle_{\tilde{k}} = k$, then it is assumed that $\langle\bm{\rho}\rangle_{\tilde{k}} \ge k+1.$
\vspace{-1.5ex}
\begin{equation}
\begin{split}
\xi(\bm{g}) &= \sum\limits_{k=1}^{\tilde{k}-1} K^{-\langle\bm{g}\rangle_k} + K^{-{\langle\bm{g}\rangle_{\tilde{k}}}}+ \sum\limits_{k=\tilde{k}+1}^K K^{-{\langle\bm{g}\rangle_k}}\\
&> \sum\limits_{k=1}^{\tilde{k}-1} K^{-\langle\bm{g}\rangle_k} + K^{-{\langle\bm{g}\rangle_{\tilde{k}}}} + (K-\tilde{k}) \cdot 0\\
&= \sum\limits_{k=1}^{\tilde{k}-1} K^{-\langle\bm{g}\rangle_k} + K^{-k}
\end{split}
\end{equation}
\vspace{-0.25ex}
\begin{equation}
\begin{split}
\xi(\bm{\rho}) &= \sum\limits_{k=1}^{\tilde{k}-1} K^{-\langle\bm{\rho}\rangle_k} + K^{-{\langle\bm{\rho}\rangle_{\tilde{k}}}}+ \sum\limits_{k=\tilde{k}+1}^K K^{-{\langle\bm{\rho}\rangle_k}}\\
&\leq \sum\limits_{k=1}^{\tilde{k}-1} K^{-\langle\bm{\rho}\rangle_k} + (K-\tilde{k}+1) \cdot K^{-{\langle\bm{\rho}\rangle_{\tilde{k}}}}\\
&\leq \sum\limits_{k=1}^{\tilde{k}-1} K^{-\langle\bm{\rho}\rangle_k} + K^{-k}
\end{split}
\vspace{-2ex}
\end{equation}

Given $\sum_{k=1}^{\tilde{k}-1} K^{-\langle\bm{g}\rangle_k} = \sum_{k=1}^{\tilde{k}-1} K^{-\langle\bm{\rho}\rangle_k},$ then $\xi(\bm{g}) > \xi(\bm{\rho})$ is proved as ture. It is obvious to obtain $\bm{g} = \bm{\rho} \Rightarrow \xi(\bm{g}) = \xi(\bm{\rho})$. Eventually, $\bm{g} \preceq \bm{\rho} \Rightarrow \xi(\bm{g}) \geq \xi(\bm{\rho})$ holds. Furthermore, the proof of $\xi(\bm{g}) \leq \xi(\bm{\rho}) \Rightarrow \bm{g} \succeq \bm{\rho}$ can be conducted by contradiction \cite{Chen2017Scheduling}.
\end{proof}

From Lemma \ref{lem:separable-obj}, it follows that\vspace{-0.5ex}
\begin{equation}
\lexmax\limits_{x_{i,j}^n \in\bm{\Theta}}\bm{\varpi}\Longleftrightarrow\minimum\limits_{x_{i,j}^n \in\bm{\Theta}} \xi(\bm{\varpi}) = \sum\limits_{n \in \mathcal{N}}\sum\limits_{i \in S_n}\sum\limits_{j \in C_i} K^{- \pi_{i,j}^n}
\vspace{-1ex}
\end{equation}\vspace{-0.25ex}
where $\xi(\bm{\varpi})$ is a summation of the term $K^{\pi_{i,j}^n}$ which is a convex function in terms of the single variable $x_{i,j}^n$. Therefore, solving the problem (\ref{eq:single-obj-substituted}) is equivalent to solving the following problem (\ref{eq:separable-obj}) with constraints (\ref{eq:customer_constraint}), (\ref{eq:provider_constraint}) and (\ref{eq:rational_def}).
\begin{equation}\label{eq:separable-obj-final}
\minimum\limits_{x_{i,j}^n \in \,\bm{\Theta}}\;\sum\limits_{n \in \mathcal{N}}\sum\limits_{i \in S_n}\sum\limits_{j \in C_i} K^{- [a_n + b_n \times (1 - \frac{Q_{i,j}}{Q^{(ref)}_n}\cdot x_{i,j}^n)]}
\vspace{-1ex}
\end{equation}
\textbf{LP Transformation.} Based on the properties of  separable convex objective and totally unimodular linear constraints holding as true, we introduce the $\lambda$-technique \cite{Meyer1976A} for optimality-equivalent Linear Programming (LP) transformation from the problem (\ref{eq:obj-lp}) in order to obtain the target service selection scheme with high efficiency. In our problem, each convex function $K^{- \pi_{i,j}^n}$ is transformed with $\lambda$-technique into another form $\psi_{i,j}^n(x_{i,j}^n)$, formulated as follows.
\vspace{-1.5ex}
\begin{equation}
\begin{split}
\psi_{i,j}^n(x_{i,j}^n) &= \sum\limits_{p \in \{0,1\}} K^{- [a_n + b_n \times (1 - \frac{Q_{i,j}}{Q^{(ref)}_n} \cdot p)]}\,\lambda_{i,j}^{n,p} \\
&= K^{- (a_n + b_n)} \,\lambda_{i,j}^{n,0} + K^{- [a_n + b_n \times (1 - \frac{Q_{i,j}}{Q^{(ref)}_n})]}\, \lambda_{i,j}^{n,1}
\end{split}
\end{equation}

The domain of decision variable $x_{i,j}^n$ is migrated from a discrete space $\{0,1\}$ to a continuous positive real space by the means of traversing each possible value $x_{i,j}^n \in \{0,1\}$, and newly introducing a couple of weighted variables $\lambda_{i,j}^{n,0}, \lambda_{i,j}^{n,1} \in \mathbb{R}^+$ subject to the constraints (\ref{eq:lamda_cons_1}) and (\ref{eq:lamda_cons_2}).
\vspace{-1.5ex}
\begin{equation}\label{eq:lamda_cons_1}
\sum\limits_{p \in \{0,1\}}\,\lambda_{i,j}^{n,p} = \lambda_{i,j}^{n,0} + \lambda_{i,j}^{n,1} = 1
\end{equation}
\vspace{-1.5ex}
\begin{equation}\label{eq:lamda_cons_2}
x_{i,j}^n\,=\,\sum\limits_{p \in \{0,1\}}\,p\cdot\lambda_{i,j}^{n,p}=\,\lambda_{i,j}^{n,1}
\vspace{-1ex}
\end{equation}

Jointly considering the linear relaxation on the integer constraints, the linear programming problem is eventually obtained as (\ref{eq:obj-lp}).
\begin{equation}\label{eq:obj-lp}
\minimum\limits_{\textbf{\emph{x},$\boldsymbol{\lambda}$}} \sum\limits_{n \in \mathcal{N}}\sum\limits_{i \in S_n}\sum\limits_{j \in C_i} K_0 \cdot\,\lambda_{i,j}^{n,0} + K_1 \cdot\, \lambda_{i,j}^{n,1}
\end{equation}

subject to,
\vspace{-1.5ex}
\begin{displaymath}
\sum\limits_{i\in S_n}\sum\limits_{j\in C_i}x_{i,j}^n = 1,\;\forall n \in \mathcal{N}
\end{displaymath}
\begin{displaymath}
\sum\limits_{n\in \mathcal{N}_i}x_{i,j}^n \le 1 \;\; \forall i \in \mathcal{M},\,\forall j \in C_i
\end{displaymath}
\begin{displaymath}
x_{i,j}^n\,=\,\lambda_{i,j}^{n,1}\, \;\;\forall n\in \mathcal{N},\, i\in S_n,\, j \in C_i
\end{displaymath}
\begin{displaymath}
\lambda_{i,j}^{n,0} + \lambda_{i,j}^{n,1} = 1\, \;\;\forall n\in \mathcal{N},\, i\in S_n,\, j \in C_i
\end{displaymath}
\begin{displaymath}
x_{i,j}^n,\,\lambda_{i,j}^{n,0},\,\lambda_{i,j}^{n,1}\,\in \mathbb{R}^+\, \;\;\forall n\in \mathcal{N},\, i\in S_n,\, j \in C_i
\end{displaymath}
\vspace{-2ex}
\begin{displaymath}
K_0=K^{- (a_n + b_n)},\,K_1=K^{- [a_n + b_n \times (1 - \frac{Q_{i,j}}{Q^{(ref)}_n})]}
\vspace{-1.5ex}
\end{displaymath}
\vspace{-2ex}
\begin{theorem}
An optimal service selection scheme derived from the problem (\ref{eq:obj-lp}) coincides with an optimal scheme derived from the problem (\ref{eq:lex_obj}).
\end{theorem}
\begin{proof}
The coefficient matrix with the property of total unimodularity guarantees the integral optimum of  the LP problem (\ref{eq:obj-lp}), which is also the optimal solution to the problem (\ref{eq:separable-obj-final}). Besides, the optimal scheme obtained from the problem (\ref{eq:separable-obj-final}) also coincides with the one to the problem (\ref{eq:single-obj-substituted}). To sum up, an optimal service selection scheme to the problem (\ref{eq:obj-lp}) is an optimal scheme to the problem (\ref{eq:single-obj-substituted}).
\end{proof}

From now on, the optimal service selection scheme across multiple service requests maximizing the lowest payment can be computed with efficient LP algorithms (e.g. Simplex Algorithm, Interior Point Method, etc.) and solvers (e.g. MOSEK \cite{Andersen2000The}, CPLEX \cite{CPLEX}, etc.).
\vspace{-0.75ex}
\section{Experiments and Results}
\vspace{-0.75ex}
\subsection{Experimental Setup}
\vspace{-0.75ex}
In our experimental setting, we respectively set the amount of service providers and service requests as 9 and 10, which is $\mathcal{M} = 9$ and $\mathcal{N} = 10$. To be more specific, there are totally 9 service providers which maintain their respective candidate services. The candidate service associated with the QoS value originates from WSDream dataset \cite{Zheng2014Investigating}, which measures response time for 5,825 types of real-world web services from disparate locations. Nine
amongst 5,825 types of web services are randomly chosen
as service providers. At the user side, 10 users simultaneously make service requests, each of which corresponds to a broker responsible for regulating the service selection process. The privilege for service selection is restricted to specified service providers.

We carry out experiments through our simulator in C++, invoking IBM CPLEX \cite{CPLEX} to solve our LP problems. The experimental simulation is conducted on Ubuntu 18.04 LTS while the processor is AMD Ryzen7 2700U 3.8GHz 64-bit with the memory size of 8GB.
\vspace{-1.5ex}
\subsection{Experimental Results}
\vspace{-0.75ex}
To extensively investigate the optimality and fairness of our proposed algorithm, we tune the providers' pricing policy (i.e. $a_n+b_n$) to evaluate the payment deviation across multiple requests and overall revenue of service providers, depicted in Fig.\ref{fig:pay_dev_compare} and Fig.\ref{fig:revenue_compare}.
The pricing policy is set as 8 levels from 1 to 8, where 1 implies the lowest pricing level and 8 is the highest. A higher pricing level indicates that the candidate service is more highly priced.

Our proposed algorithm, referred to as \emph{FASS}, is compared with two baselines - \emph{Revenue Maximization} and \emph{Randomized}. The \emph{Revenue Maximization} algorithm refers to the algorithm whose objective is to maximize the overall revenue including all of users' payments, ignoring how much respectively charged to a user individual compared with other users, while the \emph{Randomized} algorithm randomly selects a service from available candidates for execution. The following experimental results correlated to \emph{Randomized} algorithm are averaged over 1,000 runs.

On the one hand, smaller payment deviation amongst individuals  guarantees the fairness of concurrent service selection. Thanks to the notion of max-min fairness introduced, it points out in Fig.\ref{fig:pay_dev_compare} that our \emph{FASS} algorithm is at the minimum payment deviation. The \emph{Randomized} algorithm takes the second place, whereas the \emph{Revenue Maximization} algorithm performs with the maximum payment deviation, much less for fairness guarantee. On the other hand, service providers which attain higher revenue due to offering services gain higher profits. \emph{Revenue Maximization} algorithm optimizes the overall revenue from all service requests, served as the \emph{optimal} baseline in our comparison study. The \emph{Randomized} algorithm acquires the least revenue because of its blind selecting behavior. Our \emph{FASS} algorithm does not top the list, nevertheless, there simply exist tiny gaps away from the baseline of \emph{Revenue Maximization}, which is within an acceptable range. Notwithstanding  a  little  sacrifice  of  revenue gains, our \emph{FASS} algorithm achieves the fairness guarantee across multiple service requests.
\begin{figure}[!t]
    \centering
    \vspace{-0.2cm}  
    \setlength{\abovecaptionskip}{-0cm}   
    \setlength{\belowcaptionskip}{-0cm}   
    \includegraphics[width=2.5in]{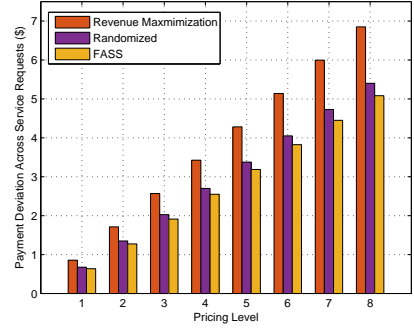}
    \caption{Payment Deviation under Different Algorithms.}
    \label{fig:pay_dev_compare}
\end{figure}

\begin{figure}[!t]
    \centering
    \vspace{-0.4cm}  
    \setlength{\abovecaptionskip}{-0cm}   
    \setlength{\belowcaptionskip}{-2cm}   
    \includegraphics[width=2.5in]{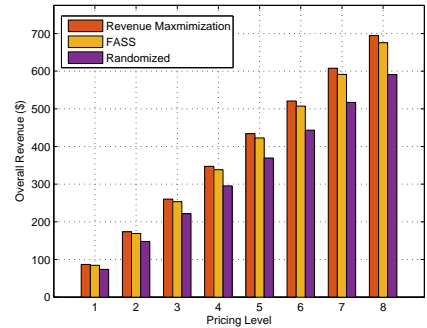}
    \caption{Overall Revenue under Different Algorithms.}
    \label{fig:revenue_compare}
\end{figure}
\begin{figure}[!t]
    \centering
    \vspace{-0.4cm}  
    \setlength{\abovecaptionskip}{-0.2cm}   
    \setlength{\belowcaptionskip}{-2cm}   
    \includegraphics[width=2.5in]{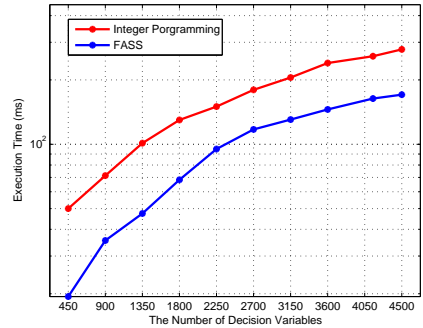}
    \caption{Algorithm Execution Time at Different Scales.}
    \label{fig:exeTm}
    \vspace{-0.6cm}
\end{figure}

Furthermore, we evaluate the practicality of our proposed algorithm by measuring the execution time of various algorithms under different problem scales. Since it is sharply resource intensive and time consuming to solve original lexicographical problem (\ref{eq:lex_obj}), the running times of integer programming (i.e. $x_{i,j}^n\, \in\, \{0,1\}$) and our FASS algorithm is elected to conduct comparative analysis.  The running times of both algorithms are demonstrated in Fig. \ref{fig:exeTm}, with the number of decision variables from 450 to 4,500. Each data point representing the execution time is averaged over 20 runs. Under the growth of problem scale (i.e. number of decision variables), the execution time of both algorithms are kept as nearly linear increase. Compared with integer programming, our FASS algorithm performs much faster over 153\% to 258\%. In the sight of numerical results, the procedure of service selection for FASS can be accomplished below tens or hundreds of
milliseconds, with the minimum of 19.40 ms for 450 variables as well as the maximum of 170.64 ms for 4,500 variables, far less than 1 s. It follows that our FASS algorithm is efficient, acceptable in practice.
\vspace{-1.25ex}
\section{Conclusions and Future Work}
\vspace{-1ex}
Fairness is an important issue in service selection when multiple users share multiple candidate services in a service ecosystem. We study the QoS-aware service selection problem from a globally fairness viewpoint, where service selection constraints are also fully considered. With the objective of achieving max-min fairness across the entire system, we formulate the service selection as a lexicographical maximization problem. An efficient algorithm is designed to solve such problem with acceptably low overhead by introducing $\lambda$-technique and linear relaxation. Our proposed approach are validated by theoretical analysis and experimental results based on real-world dataset.

There are several avenues for future work. On the one hand, dynamic service composition scheme might be designed based on the basic idea proposed in this paper. Since static optimization is already far from trivial to deal with practically due to its inherently multi-objective and discrete nature, there are several hard problems to be addressed especially for the performance issue. On the other hand, another avenue of our future work is to deploy our approach in real-life environments such as cloud computing, web services and mobile service systems. Experimental results obtained from reality should provide us with more insights of the user/system behaviors and algorithm optimization. Also, pricing schemes and gaming among service providers/users in real-world systems are interesting problems for researchers in this community to study.

\section*{Acknowledgment}
This work is supported by the National Key Research and Development Program of China (Nos. 2018YFB1003804 and 2016YFC0303707), the National Natural Science Foundation of China
(No.61772479), the Fundamental Research Funds for the Central Universities (No. 2462018YJRC040), the Natural Science Foundation of Shandong Province of China for Major Basic Research Projects (No. ZR2017ZB0419), and the TaiShan Industrial Experts Program of Shandong Province of China (No. tscy20150305).



%


\vspace{-1.5ex}
\newcommand{\BIBdecl}{\setlength{\itemsep}{0.001 em}}
\bibliographystyle{IEEEtran}
\bibliography{./ICWS2019}

\end{document}